\documentclass[letterpaper, 10 pt, conference]{ieeeconf}  

\IEEEoverridecommandlockouts                              

\usepackage{url}
\usepackage{cite}
\usepackage{amsmath,amssymb,amsfonts}
\usepackage{algorithmic}
\usepackage{graphicx}
\usepackage{subcaption}
\usepackage{textcomp}
\usepackage{xcolor}
\def\BibTeX{{\rm B\kern-.05em{\sc i\kern-.025em b}\kern-.08em
    T\kern-.1667em\lower.7ex\hbox{E}\kern-.125emX}}

\newtheorem{theorem}{Theorem}

\newtheorem{lemma}{Lemma}
\newtheorem{remark}{Remark}

\usepackage{color}

\usepackage{algorithm}
\usepackage{algorithmic}

\title{\LARGE \bf
    Stability-Preserving Model Reduction of Networked Lur'e Systems
}

\author{Yangming Dou, Xiaodong Cheng, and Jacquelien M. A. Scherpen
\thanks{Y. Dou and J.M.A Scherpen are with Jan C. Willems Center for Systems and Control, ENTEG, Faculty of Science and Engineering, University of Groningen, 9747 AG Groningen, The Netherlands 
{\tt\small \{y.dou, j.m.a.scherpen\}@rug.nl}}%
\thanks{X. Cheng is with Mathematical and Statistical Methods Group (Biometris), Department of Plant Science, Wageningen University \& Research, 6700 AA Wageningen, The Netherlands.
{\tt\small xiaodong.cheng@wur.nl}}%
}

\begin{document}
 
\maketitle
\thispagestyle{empty}
\pagestyle{empty}

\begin{abstract}
    This paper proposes a model reduction approach for simplifying the interconnection topology of Lur'e network systems. A class of reduced-order models are generated by the projection framework based on graph clustering, which not only preserve the network structure but also ensure absolute stability. Furthermore, we provide an upper bound on the input-output approximation error between the original and reduced-order Lur'e network systems, which is expressed as a function of the characteristic matrix of graph clustering. Finally, the results are illustrated via a numerical example.
\end{abstract}

\begin{keywords}
    Lur'e system, model reduction, absolute stability, graph clustering
\end{keywords}

\section{Introduction}
Lur'e systems represent an important class of nonlinear systems that consist of linear dynamics with a static nonlinear feedback.
Many physical nonlinear systems, such as mechanical systems, power systems, and hyperchaotic attractors can be modeled as Lur'e systems, see some examples in \cite{Besselink2009Lure,lee2010synchronization}. 
The interconnection of  Lur'e systems  gives rise to  networked Lur'e systems, which gain much attention from the literature, see e.g.,\cite{suykens1999robust,murray2007recent,zhang2015distributed}. 

From the perspective of design and optimization, dealing with large-scale models of dynamic systems can often be challenging. Therefore, model reduction serves as an indispensable tool for generating lower-order approximations that facilitate efficient design and optimization processes.
Various methods have been developed for the model reduction of Lur'e systems.  \cite{zhang2015dissipativity}  introduces the dissipativity-based model reduction for Markov jump Lur’e systems using linear matrix inequalities.\cite{Besselink2009Lure} proposes a balanced truncation approach for Lur'e systems, which preserves the absolute stability of the reduced-order model, while  \cite{ChengEJC201911} presents a generalized balanced truncation approach  for Lur'e networks, focusing on preserving the synchronization property of the networks.  
In the context of Lur'e network systems, how to reduce the number of interconnected subsystems (or nodes) is also a crucial research problem.  In \cite{deng2014structure}, a method based on Kullback-Leibler divergence is applied to simplify the interconnected structure of Lur'e networks. However, the model reduction error is difficult to characterize since only the linear part is considered for the divergence analysis, disregarding the nonlinear component of the system.

In general, reducing a system containing nonlinear elements is challenging, especially when dealing with nonlinear systems with network structure. Nevertheless, insights can be drawn from model reduction techniques employed in linear network systems. In such systems,   a mainstream method for reducing complicated network structures is graph clustering, which is realized by dividing the subsystems into several disjoint clusters \cite{Cheng2021AR,monshizadeh2014clusterting}. Different from  conventional model reduction methods such as balanced truncation \cite{antoulas2005approximation} or moment matching \cite{astolfi2010model}, the reduced order system obtained by clustering can still be represented as a network system with fewer number of nodes than its full-order counter part. How to find suitable clusters becomes the key in this kind of methods. \cite{besselink2015clustering} proposes a clustering-based model reduction method for networked passive systems by analyzing the controllability and observability properties
of associated edge systems. \cite{ishizaki2013model} introduces the notion of clustering reducibility, which is related to the approximation error. The works in \cite{ChengTAC20172OROM,ChengTAC2018MAS,Cheng2019digraph} present a dissimilarity-based clustering approach for both undirected and directed network systems, where the dissimilarity between two nodes can be featured in the difference of node behaviors with respect to external inputs. In \cite{martin2018large}, clustering method is generalized to scale-free networks, where the reduced system is obtained by minimizing the scale-free cost function.
However, compared to the reduction of linear network systems, model reduction for nonlinear network systems is remain relatively underdeveloped.

In this paper, we propose a clustering-based model reduction method for reducing the network structure of the Lur'e network, meanwhile preserving the absolute stability. The main framework is based on the clustering-based project using the characteristic matrix of a clustering. Different from \cite{deng2014structure}, we consider both the linear and nonlinear parts of Lur'e networks in the framework, i,e, both parts will determine the upper hound of the approximation error. Moreover, the clustering-based projection framework guarantees the preservation of the absolute stability in reduced-order Lur'e networks. In this paper, 
an explicit expression of the upper bound on input-to-output error between the original and reduced-order network systems is provided, where the bound can be calculated via a linear system parameterized by the characteristic matrix of a graph clustering. 

The remainder of this paper is structured as follows. The preliminaries and problem setting are introduced in Section~\ref{sec:prelim}. Section \ref{sec:main} then provides the clustering-based model reduction method for Lur'e networks and analyze the bound on the reduction error. In Section \ref{sec:simulation}, an example is shown to illustrate the results. Finally, in Section \ref{sec:conclusion}, the conclusion remarks and potential future works are given.
 
 \textit{Notation:} The symbol $\mathbb{R}$ denotes the set of real numbers, and $\mathbb{R}_+$ denotes the set of nonnegative real numbers. $I_n$ represents the identity matrix of size $n$. $\mathbf{1}$ represents a vector with all the elements equal to $1$. 
$\|x(t)\|_2$ denotes the $\mathcal{L}_2$-norm of a signal $x(t)$, and $\|G(s)\|_{\mathcal{H}_\infty}$ denotes the $\mathcal{H}_{\infty}$-norm of the transfer function $G(s)$ of a linear time-invariant system.
\section{Preliminaries and Problem Setting}
\label{sec:prelim}


In this paper, we consider a Lur'e network system composed of $N$ subsystems formulated as follows:
\begin{equation}\label{SYS scalar}
    \Sigma_i:\begin{aligned}
        &\Dot{x}_i=-a_ix_i+u_i-\phi(x_i),
        ~i=1,2,\cdots, N,
    \end{aligned}
\end{equation}
where $x_i\in \mathbb{R}$ and $u_i\in \mathbb{R}$ are the state and the input of the $i$th subsystem $\Sigma_i$. $v_i\in \mathbb{R}$ is the feedback. $\phi(x_i) \in \mathbb{R}\mapsto \mathbb{R}$ is a continuous nonlinear function. Throughout the paper, we 
assume that the nonlinearity $\phi(x_i)$ in each subsystem $i$ is \textit{slope-restricted} as
	\begin{equation} \label{assu:eqsecbo}
	0 \leq \dfrac{\phi(x_a) - \phi(x_b)}{x_a - x_b} \leq \mu_i,
	\end{equation}
 for all $x_a, x_b \in \mathbb{R}$ and $x_a \ne x_b$, where $\phi(0) = 0$, $\mu_i > 0$ is known.

Define the input to each subsystem as
\begin{equation}
 u_i=\Sigma^N_{j=1}w_{ij}(x_j-x_i)+b_{ik}u_{ei},   
\end{equation}
where $w_{ij}$ is the weight of the edge between the two nodes $i$ and $j$, and $u_{ek}$ is the external input with the gain $b_{ik} \in \mathbb{R}$, which is $0$ if $k$-th input does not affect node $i$. Thus the compact networked system  can be expressed as

\begin{equation}
\label{lure net}
    \Sigma:\begin{aligned}
        &\Dot{x}=A_Lx+Bu_e-\Phi(x)\\
    \end{aligned}
\end{equation}
where $x=[x_1, x_2, \cdots, x_N]^T$, $u_e=[u_{e1}, u_{e2}, \cdots, u_{ep}]^T$, and $\Phi(x)=[\phi(x_1), \phi(x_2), \cdots, \phi(x_N)]^T$,  with $p$ the number of external inputs, and  $B \in \mathbb{R}^{N \times p}$ is a matrix with $b_{ik}$ as its $(i,k)$-th entry. 
\begin{align}\label{AL}
    A_L : = -A - L,
\end{align}
where
$A=\text{diag}(a_1, a_2, \cdots ,a_N)$, and $L$ is the Laplace matrix of graph $\mathcal{G}$, and its each element is defined as
\begin{equation}
L_{ij}=\left\{\begin{aligned}
 &-w_{ij}, ~~~~~~~~i\neq j\\
   & \sum^N_{j=1,j\neq i}w_{ij},~~\text{otherwise}.
\end{aligned}\right.   
\end{equation}
Clearly, $A_L$ is symmetric and negative-definite.


This paper aims for structure-preserving model reduction for diffusively coupled Lur'e networks in the form of \eqref{lure net}, and the reduced-order model not only approximates the input-output behavior of the original network system with a certain accuracy but also inherits an interconnection structure with diffusive couplings. 
Specifically, the problem addressed in this paper is as follows.

Given a networked Lur'e system \eqref{lure net}, our objective is to derive a simplified model described by:
\begin{equation}
\label{lurerednet}
\Hat{\Sigma}: 
 \Dot{z}=\Hat{A}_Lz+\Hat{B}u_e- \hat\Phi(\Hat{x}), \quad
 \Hat{x}=\Pi z,
\end{equation}
where $z\in \mathbb{R}^r$ denotes the state of the reduced-order system, and $\hat{x}\in \mathbb{R}^N$ represents an approximation of $x$. The matrix $\Hat{A}_L$ can be composed of a diagonal matrix $\hat{A}$ and a Laplacian matrix $\hat{L}$  characterizing a reduced undirected graph. Additionally, the reduction error $\lVert x(t) - \hat{x}(t) \rVert_{2}$ remains sufficiently small relative to the external input $\lVert u(t) \rVert_{2}$.

\section{Main Results}
\label{sec:main}
This section presents the main results of this paper, where we first show the absolute stability can be preserved by the proposed method and then analyze the upper bound of the approximation error. 

\subsection{Absolute Stability}

To guarantee the approximation error $\| x(t) - \hat{x}(t)\|_2$ to be up-bounded w.r.t. the external inputs $u(t)$, we need to study the condition under which both the original network system \eqref{lure net} and the reduced-order system \eqref{lurerednet} are stable. Specifically, we consider the concept of absolute stability, which essentially means that the origin of a Lur'e system is globally uniformly asymptotically stable for any nonlinearity in the given sector.

Note that if each subsystem has a slope-restricted nonlinearity as in \eqref{assu:eqsecbo}, then it is not hard to show that the nonlinearity $\Phi(x)$ in the Lur'e network system will satisfy the
incrementally sector-bounded condition described as: 
\begin{equation}\label{incrementally sector bounded}
 [\Phi(x_a)-\Phi(x_b)]^T[\Phi(x_a)-\Phi(x_b)-K_\mu (x_a-x_b)]\leq 0   
\end{equation}
for all $x_a,~x_b \in \mathbb{R}^N$, $x_a\neq x_b$
and $\Phi(0)=0$, where 
$$K_\mu=\text{diag}(\mu_1, \mu_2,\cdots, \mu_N).$$


Based on that we provide a sufficient condition for the absolute stability of the Lur'e network system \eqref{lure net}.
\begin{lemma}
    The unforced original system is absolutely stable if there exist a positive definite symmetric matrix $P$, a symmetric matrix $W$ and a positive constant $\epsilon$ such that
    \begin{subequations}
    \label{eq:abs}
            \begin{align}\label{eq: abs1}
        PA_L+A^T_LP&=-W^TW-\epsilon P\\
        P&=K_\mu-\sqrt{2}W. \label{eq: abs2}
    \end{align}
    \end{subequations} 
\end{lemma}
\begin{proof}
   Let the external input $u_e=0$, then the Lur'e network system becomes
\begin{equation}
    \label{eq unforced Lure}
         \Dot{x}=A_Lx-\Phi(y), \quad y=x
\end{equation} 
From Lemma~10.3 in \cite{Khalil1996Noninear}, the system \eqref{eq unforced Lure} is absolutely stable if $Z(s)=I_N+K_\mu(sI_N-A_L)^{-1}$ is strictly positive real. Apparently, $(A_L,I_N)$ is controllable and $(A_L,K_\mu)$ is observable, then according to Lemma~10.3 in \cite{Khalil1996Noninear}, $Z(s)=I_N+K_\mu(sI_N-A_L)^{-1}$ is strictly positive real if and only if \eqref{eq: abs1} and \eqref{eq: abs2} satisfy. 
\end{proof}

In the following, we assume that the sector bound $K_\mu$ of the original Lur'e network system \eqref{lure net} satisfies \eqref{eq:abs}. Then, in the next section, we study how to generate a reduced-order model \eqref{lurerednet} such that it can be interpreted as a reduced network system, and moreover, the absolute stability is retained.

\subsection{Clustering-based Model Reduction}

Clustering-based methods is a well-studied model reduction method for simplifying network systems \cite{Cheng2021AR}, allowing the reduced model to be interpreted as a reduced network, where each node in the reduced network corresponds to a cluster of nodes in the original network. 

Graph clustering is to partition all the nodes in a graph into several nonempty and disjoint subsets. 
A graph clustering  can be characterized by a binary matrix $\Pi \in  \{0,1\}^{N \times r}$, whose element is defined as
\begin{equation*}
\Pi_{ij} :=  \begin{cases}
1 & \text{if node}~i~\text{is included in clustering} ~j,\\
0 & \text{otherwise}. \\ 
\end{cases}
\end{equation*}
Since each node can only belong to one cluster, we have $\Pi \mathbf{1}_N = \mathbf{1}_r$, where $N$ and $r$ represent the numbers of nodes and clusters, respectively.
Given the reduced order $r$, clustering-based model reduction is to find a graph clustering with $r$ clusters of nodes such that $x\approx \hat{x} =  \Pi z$, where $\Pi \in \mathbb{R}^{N \times r}$ is the characteristic matrix of the clustering. Here, $\hat{x}_i$ represents an approximation of the collective behavior of all the nodes in cluster $i$. 

In this paper, we also resort to graph clustering to reduce the networked Lur'e system \eqref{lure net} for preserving the network structure. Let $\Pi \in \mathbb{R}^{N \times r}$ be the characteristic matrix of a clustering of the underlying graph, and denote: 
\begin{align}
    \Pi ^\dagger: =(\Pi^T K_\mu \Pi)^{-1}\Pi ^T K_\mu
\end{align}
 with $K_\mu$ in \eqref{incrementally sector bounded}, such that $\Pi ^\dagger \Pi = I_r$. Then the coefficient matrices in the reduced-order model \eqref{lurerednet} are given by
\begin{align}
\label{eq:redmat}
    \Hat{A}_L=\Pi ^{\dagger}A_L\Pi, \ \Hat{B}=\Pi ^\dagger B, \ \hat\Phi(\cdot) =  \Pi ^\dagger \Phi(\cdot).
\end{align}

First, we show that the reduced-order network model in the form of \eqref{lurerednet} preserves not only the network structure but also the absolute stability.
 
\begin{theorem}
The reduced-order Lur'e network system in \eqref{lurerednet} is absolutely stable and preserves the network structure such that 
$\Hat{A}_L $ is the sum of a positive diagonal matrix and a reduced-dimension Laplacian matrix.
\end{theorem}
\begin{proof}
For the reduced-order Lur'e network system \eqref{eq unforced lurerednet}, its unforced system is
\begin{equation} \label{eq unforced lurerednet}
        \Dot{z}=\Hat{A}_Lz-\Pi ^{\dagger}\Phi(\Hat{x}), \quad
        \Hat{x}=\Pi z
\end{equation}
First, we need to prove $(\Hat{A}_L, \Pi^\dagger)$ is controllable, $(\hat A_L,\Pi)$ is observable.
The controllability and observability can be seen from the controllability and observability
matrices 
\begin{align*}
    \hat{\mathcal{C}}&=\begin{bmatrix}\Pi^\dagger&\Hat{A}_L\Pi^\dagger&\Hat{A}_L^2\Pi^\dagger&\cdots&\Hat{A}_L^{r-1}\Pi^\dagger\end{bmatrix}, \\
  \hat{\mathcal{O}}&=\begin{bmatrix}
        \Pi&(\Pi \Hat{A}_L)^T&(\Pi \Hat{A}_L^2)^T&\cdots&(\Pi \Hat{A}_L^{r-1})^T
    \end{bmatrix}^T,
\end{align*}
which are full rank due to $\text{rank}(\Pi)=\text{rank}(\Pi^\dagger)=r$.

Following the similar reasoning as the proof of Lemma 1, to show the  absolute stability of the above
system, we need to prove $\hat{Z}(s)=I_N+K_\mu \Pi(sI_r-\hat{A}_L)^{-1}\Pi^\dagger$ is strictly positive real. Notice that $(\hat A_L,K_\mu \Pi)$ is observable since $(\hat A_L,\Pi)$ is observable and $K_\mu$ is nonsingular, By Lemma 10.2 in \cite{Khalil1996Noninear}, $\hat{Z}(s)=I_N+K_\mu \Pi(sI_r-\hat{A}_L)^{-1}\Pi^\dagger$ is positive real if and only if there exist a positive definite symmetric matrix $\hat{P}$, matrix $\hat{W}$ and a positive constant $\epsilon$ such that
\begin{subequations}
    \begin{align}
    \hat{P}\hat{A}_L+\hat{A}_L^T\hat{P}&=-\hat{W}^T\hat{W}-\epsilon\hat{P}\label{eq: re abs} \\ 
    \hat{P}\Pi^\dagger&=\Pi^TK_\mu-\sqrt{2}\hat{W}^T \label{eq: re abs2}
\end{align}
\end{subequations}

From \eqref{eq: re abs}, we have
\begin{align}\label{eq: re abs3}
  \hat{P}\Pi^\dagger A_L \Pi+\Pi^T A_L^T(\Pi^\dagger)^T\hat{P}&=-\hat{W}^T\hat{W}-\epsilon\hat{P}  
\end{align}
By \eqref{eq: re abs2} and \eqref{eq: re abs3}, we have
\begin{align}\label{eq: re abs P}
  \hat{P}=\Pi^TK_\mu\Pi-\sqrt{2}\hat{W}^T \Pi 
\end{align}
and
\begin{align}\label{eq: re abs4}
    &(\Pi^TK_\mu-\sqrt{2}\hat{W}^T) A_L \Pi+\Pi^T A_L^T(\Pi^TK_\mu-\sqrt{2}\hat{W}^T)^T \nonumber\\
  &=-\hat{W}^T\hat{W}-\epsilon(\Pi^TK_\mu\Pi-\sqrt{2}\hat{W}^T \Pi )  
\end{align}
If we choose $\hat{W}=W\Pi$, from \eqref{eq: re abs4}, we have 
\begin{align}
   &\Pi^T(K_\mu-\sqrt{2} W) A_L \Pi+\Pi^T A_L^T(K_\mu-\sqrt{2} W)\Pi \nonumber\\
  &=-\Pi^T W^TW\Pi-\epsilon(\Pi^TK_\mu\Pi-\sqrt{2}\Pi^T W \Pi ), 
\end{align}
From \eqref{eq:abs}, we can obtain
\begin{align}
   &\Pi^T P A_L \Pi+\Pi^T A_L^T P\Pi =-\Pi^T W^TW\Pi-\epsilon \Pi^TP\Pi. 
\end{align}
Thus \eqref{eq: re abs} and \eqref{eq: re abs2} hold if \eqref{eq:abs} satisfies, then $\hat{Z}(s)=I_N+K_\mu \Pi(sI_r-\hat{A}_L)^{-1}\Pi^\dagger$ is strictly positive real, which implies the unforced reduced-order system is also absolutely stable.
\end{proof}

Theorem 1 shows that by using the specific formulation of reduced-order matrices in \eqref{eq:redmat},  the reduced-order Lur'e network system \eqref{lurerednet} is guaranteed to be absolutely stable,   regardless of the choice of the characteristic matrix $\Pi$. Then, in the following section, we will study how the matrix $\Pi$ can affect the approximation error between the original and reduced-order systems.

\subsection{Analysis of Reduction Error}\label{AA}

In this section, we derive an upper bound on the reduction error between the original and reduced-order Lur'e networks, represented as a function of $\Pi$, i.e., the characteristic matrix of graph clustering.

Before proceeding, the following lemma first presents the upper bounds on the approximation errors on the linear parts.
\begin{lemma}\label{Lemma3}
Denote the following transfer matrices for the linear parts in the original and reduced-order systems:
\begin{align}
\label{eq:tfs}
    & g_{u_e}(s) =(sI_N-A_L)^{-1}B, \ 
     g_v(s) =(sI_N-A_L)^{-1},  \nonumber 
     \\ 
    & \hat{g}_{u_e}(s)  =\Pi(sI_r-\hat{A}_L)^{-1}\hat{B},  \ 
  \hat{g}_v(s)   =\Pi(sI_r-\hat{A}_L)^{-1}\Pi^{\dagger}.
\end{align}
Assume \eqref{eq:abs} holds, then the following error bounds hold:
 \begin{align}
     \label{e_gue}
    \|g_{u_e}(s)-\hat{g}_{u_e}(s)\|_{\mathcal{H}_\infty}\leq \gamma_H\|g_{u_e}(s)\|_{\mathcal{H}_\infty},
\\
\label{e_gv}
    \|g_v(s)-\hat{g}_v(s)\|_{\mathcal{H}_\infty}\leq \gamma_H\|g_v(s)\|_{\mathcal{H}_\infty},
 \end{align}
where $\gamma_H$ is the $\mathcal{H}_\infty$ norm of the following linear system
\begin{equation}\label{eq hs}
  H(s, \Pi)=C_H(sI_r-A_H)^{-1}B_H+D_H,
\end{equation}
with $A_H=\Pi^{\dagger}A_L\Pi$, $B_H=\Pi^{\dagger}A_L(I_N-\Pi \Pi^{\dagger})$, $C_H=\Pi$, and $D_H=I_N-\Pi \Pi^{\dagger}$.
$\hfill{} \blacksquare$
\end{lemma}

The proof of {Lemma \ref{Lemma3}} is shown in Appendix~\ref{ap:lin_bound}. Since $\Pi^\dagger A_L\Pi=(\Pi ^TK_\mu \Pi)^{-1}\Pi ^TK_\mu A_L\Pi$, which is Hurwitz for any full rank characteristic matrix $\Pi$, the $\mathcal{H}_\infty$ norm of $H(s)$ always exists.  Based on Lemma~\ref{Lemma3}, we present the approximation error bound in the following. 

  \begin{theorem}
  \label{thm:errbound}
Assume \eqref{eq:abs} holds. If $\|g_v(s)\|_{\mathcal{H}_\infty}<\frac{1}{(\gamma_H+1)\mu_{\text{max}}}$, the $\mathcal{H}_{\infty}$ norm of the approximation error is bounded by
\begin{equation}\label{errorbound}
    \|x(t)-\hat{x}(t)\|_2 
  \leq \Gamma(\gamma_H) \|u_e(t)\|_2,
\end{equation}
where
\begin{equation}\label{Gamma}
    \Gamma(\gamma_H)=\frac{\gamma_H\kappa_{u_e}}{\left[1-(\gamma_H+1)\mu_{\text{max}} \kappa_v\right](1-\mu_\text{max}\kappa_v)},
\end{equation}
and $\kappa_{u_e}=\|g_{u_e}(s)\|_{\mathcal{H}_\infty}$, $\kappa_v=\|g_v(s)\|_{\mathcal{H}_\infty}$.
$\hfill{} \blacksquare$
\end{theorem}

The proof of Theorem~\ref{thm:errbound} is provided in the Appendix~\ref{ap:bound}. 

\begin{remark}
    The assumption of the theorem $\kappa_v=\|g_v(s)\|_{\mathcal{H}_\infty}<\frac{1}{(\gamma_H+1)\mu_{\text{max}}}$ implies  that the two terms in the denominator of $\Gamma(\gamma_H)$, i.e. $1-(\gamma_H+1)\mu_{\text{max}}\kappa_v $ and $ (1-\mu_{\text{max}}\kappa_v)$, are both positive. Hence the error bound \eqref{errorbound} is well-defined. 
\end{remark}

Observe that in \eqref{Gamma}, only $\gamma_H$ is dependent on the choice of $\Pi$, or equivalently, graph clustering of the original network. The other parameters, $\mu_{\text{max}}$, $\kappa_{u_e}$ and $\kappa_v$, are priori since they are determined by the original network system.

Furthermore, it can be verified that $\Gamma(\gamma_H)$ is a monotonically increasing function with respect to $\gamma_H$, i.e. a smaller $\gamma_H$ will lead to lower $\Gamma(\gamma_H)$. As $\gamma_H$ is the $\mathcal{H}_\infty$-norm of the linear system $H(s)$, we can use a Riccati inequality or an LMI to characterize $\mathcal{H}_\infty$, then $\Pi$ can be selected to minimize $\gamma_H$. This would also lead to a lower error bound on the approximation the nonlinear Lur'e network.   
In particular, if $\Pi = I_N$, then $\gamma_H=0$, which yields $\Gamma(\gamma_H)=0$, meaning that the reduced-order model has exactly the same outputs as the original system with the same external inputs applied.

It is also worth mentioning that since the dimension of $A_H$ in \eqref{eq hs} is $r$, that is the dimension of the reduced-order system, to obtain $\gamma_H$ does not require an expensive computation. Therefore it will be beneficial for the subsequent  optimization procedure that is to find an optimal $\Pi$ to minimize $\gamma_H$. However, we leave the detailed discussion to our future work.


		
		
		
		
		
	  		
	  		
	  		
		
		
	 
		

\section{Simulation Results}
\label{sec:simulation}
To illustrate the proposed model reduction approach for Lur'e networks, we consider a netwrok example of   100 nodes which is shown as Fig. \ref{fig: original network}. The network is generated by the B-A Scale-Free Network Generation algorithm  \cite{George2023ba}. The 100 nodes are divided into 7 clusters: nodes 1, 2, 3-22, 23-42, 43-62, 63-81, and 82-100. The reduced-order network resulting from the given clustering is shown in Fig.~\ref{fig: reduced-order network}. Note that the reduced graph is now bidirectional, but it is not undirected, as $\Pi^\dagger L \Pi$ is no longer symmetric. 

For  the original system in \eqref{lure net} and \eqref{AL}, we choose $A=2I_{100}$, $F=[I_2, 0_{98\times 98}]^T$. The nonlinearity of each subsystem is
\[ \phi(x_i)=|x_i+0.1|-|x_i-0.1|,\] 
thus $K_\mu=0.2I_{100}$.

First, we show both the unforced original system and the reduced-order system are absolutely stable. Set the external input $u_e=0$ and choose random     values ranging from $-1$ to $1$ as initial states of both the original and reduced-order systems. From Fig. \ref{fig:x} and Fig. \ref{fig:z}, we observe that  the states of both the original  and the reduced-order Lur'e network systems asymptotically converge to $0$, which implies the stability.

Then we choose both the external inputs as $\sin(t)$, and set the initial states of both systems are zero. The state trajectories are plotted in Fig. \ref{fig:y} and Fig. \ref{fig:yr}, respectively, where the nodes in the same cluster are indicated by the same color. Note that nodes 1 and 2 form two clusters, and the approximation errors are shown to be relative small. In contrast, the other clusters are formed without any optimization, leading to larger approximation errors. 

To validate the error bound in Theorem~\ref{thm:errbound}, we estimate the input-to-output error of the model reduction as follows:
\begin{align*}
    \gamma_\epsilon^2 = \frac{\int_{0}^{T} x^T(t)x(t)dt}{\int_{0}^{T} u^T(t)u(t)dt} \approx   \frac{\sum_{k=0}^{T/\delta t} x^T(k)x(k)\delta t}{\sum_{k=0}^{T/\delta t} u^T(k)u(k) \delta t},
\end{align*}
where $T$ is the length of the simulation time, and $\delta t$ is the stepsize. In this simulation, we obtain $\gamma_\epsilon = 0.0761$. Meanwhile, using Lemma~\ref{Lemma3}, we compute 
$\|g_{u_e}(s)\|_{\mathcal{H}_\infty}=0.1372$, $\|g_v(s)\|_{\mathcal{H}_\infty}=0.5$, and $\gamma_H=1.2607$. Therefore, $\|g_v(s)\|_{\mathcal{H}_\infty}<\frac{1}{(\gamma_H+1)\mu_{\text{max}}}$ holds. It
then leads to the error bound $\Gamma(\gamma_H)=0.2484$, according to Theorem~\ref{thm:errbound}.

		\begin{figure}[!tp]
			\begin{minipage}[t]{0.6\linewidth}
				\centering
				\includegraphics[width=\textwidth]{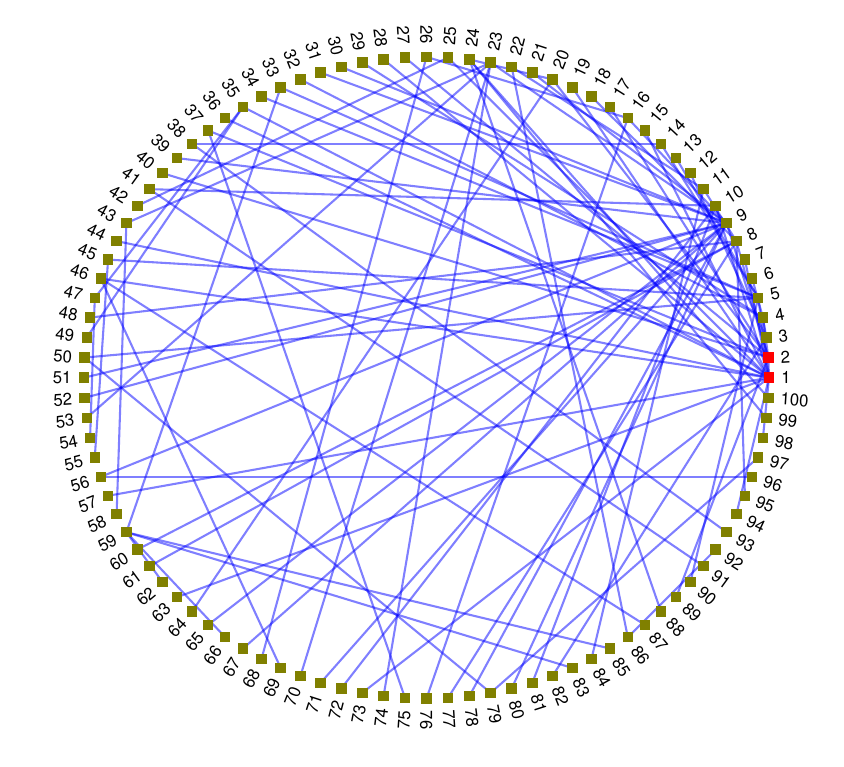}    
				\subcaption{}
				 \label{fig: original network}
			\end{minipage}%
			\begin{minipage}[t]{0.4\linewidth}
				\centering
		\includegraphics[width=\textwidth]{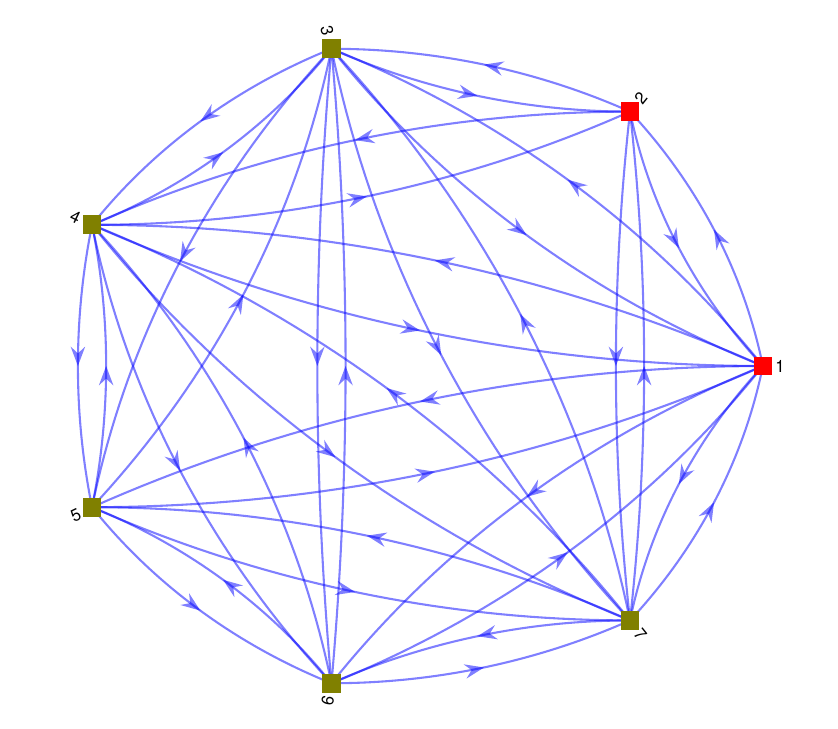}    
				\subcaption{}
				\label{fig: reduced-order network}
			\end{minipage}%
			\caption{(a) The topology of the original network. (b) The topology of the reduced network. The input nodes are highlighted by the red color.}
		\end{figure} 
 
		\begin{figure}[!tp]
			\begin{minipage}[t]{0.5\linewidth}
				\centering
				\includegraphics[width=\textwidth]{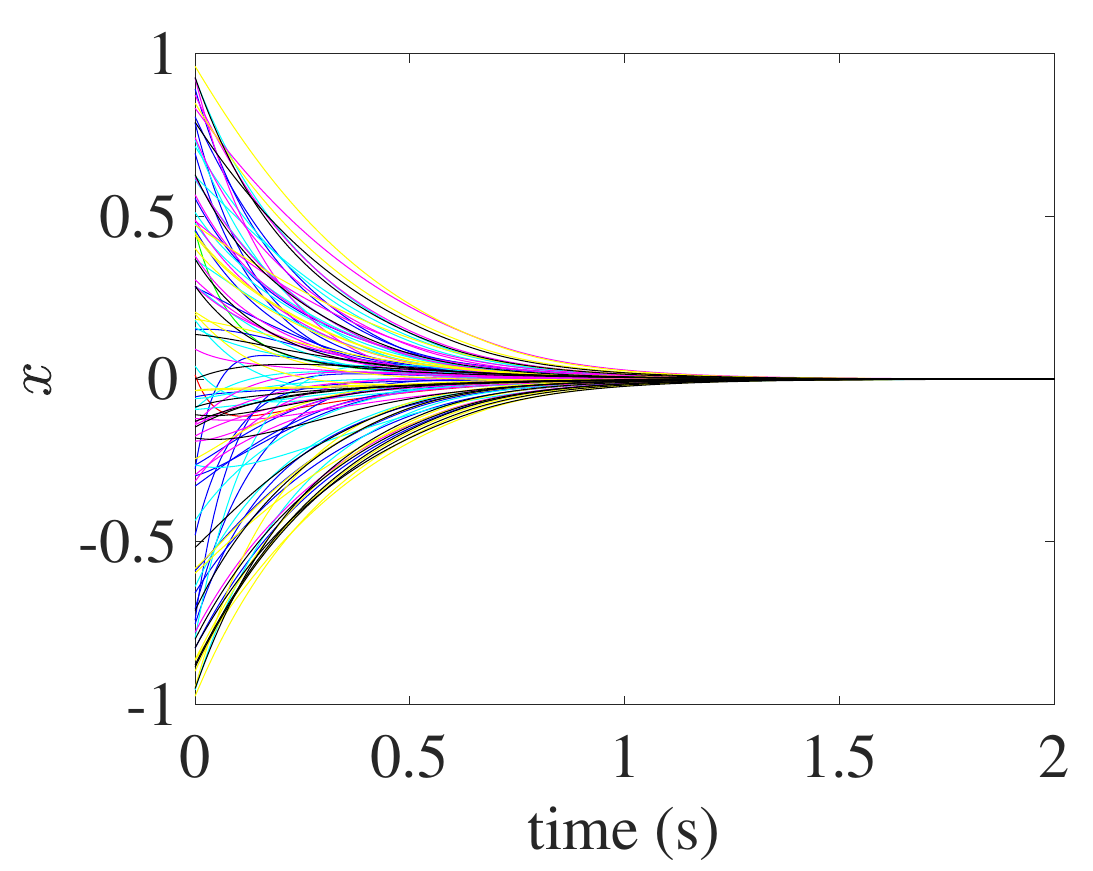}    
				\subcaption{}
				\label{fig:x}
			\end{minipage}%
			\begin{minipage}[t]{0.5\linewidth}
				\centering
				\includegraphics[width=\textwidth]{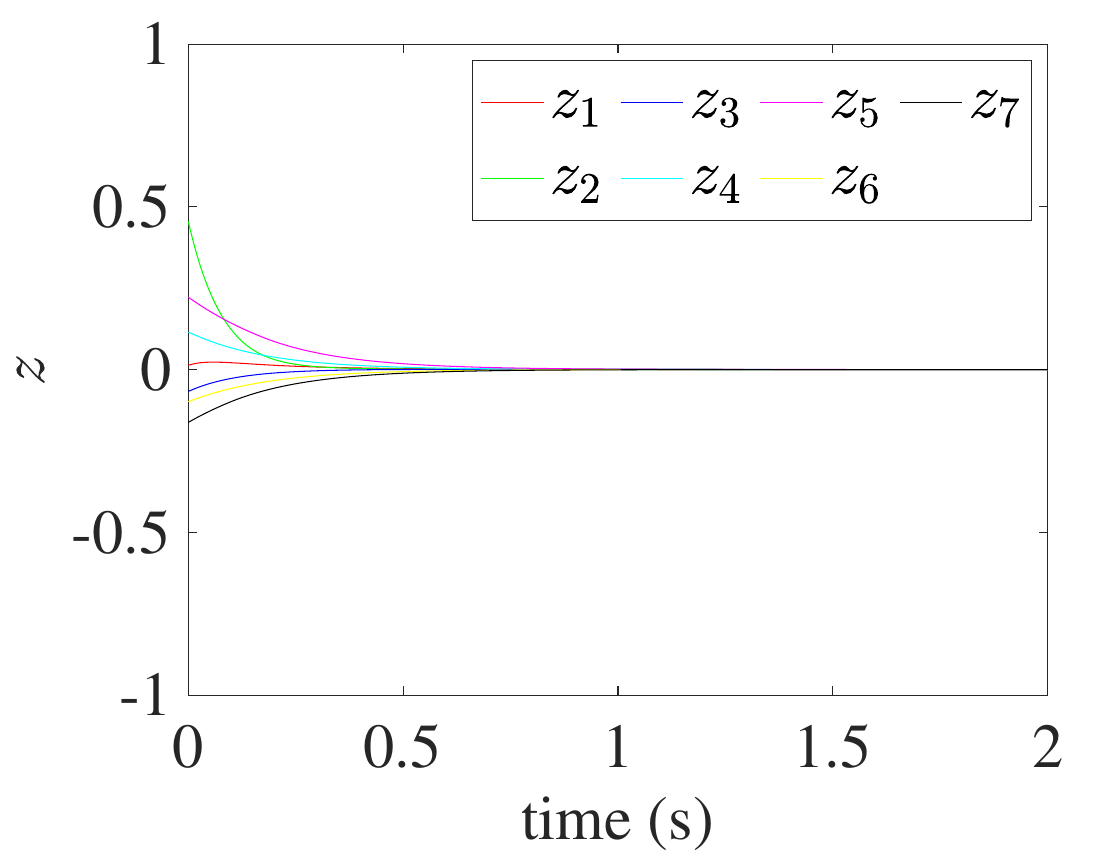}    
				\subcaption{}
				\label{fig:z}
			\end{minipage}%
			\caption{(a) State trajectories of the original Lur'e network system with random initial states and $u_e = 0$. (b)  State trajectories of the reduced-order  system with random initial states and  $u_e = 0$.}
		\end{figure} 
		\begin{figure}[!t]
			\begin{minipage}[t]{0.5\linewidth}
				\centering
				\includegraphics[width=\textwidth]{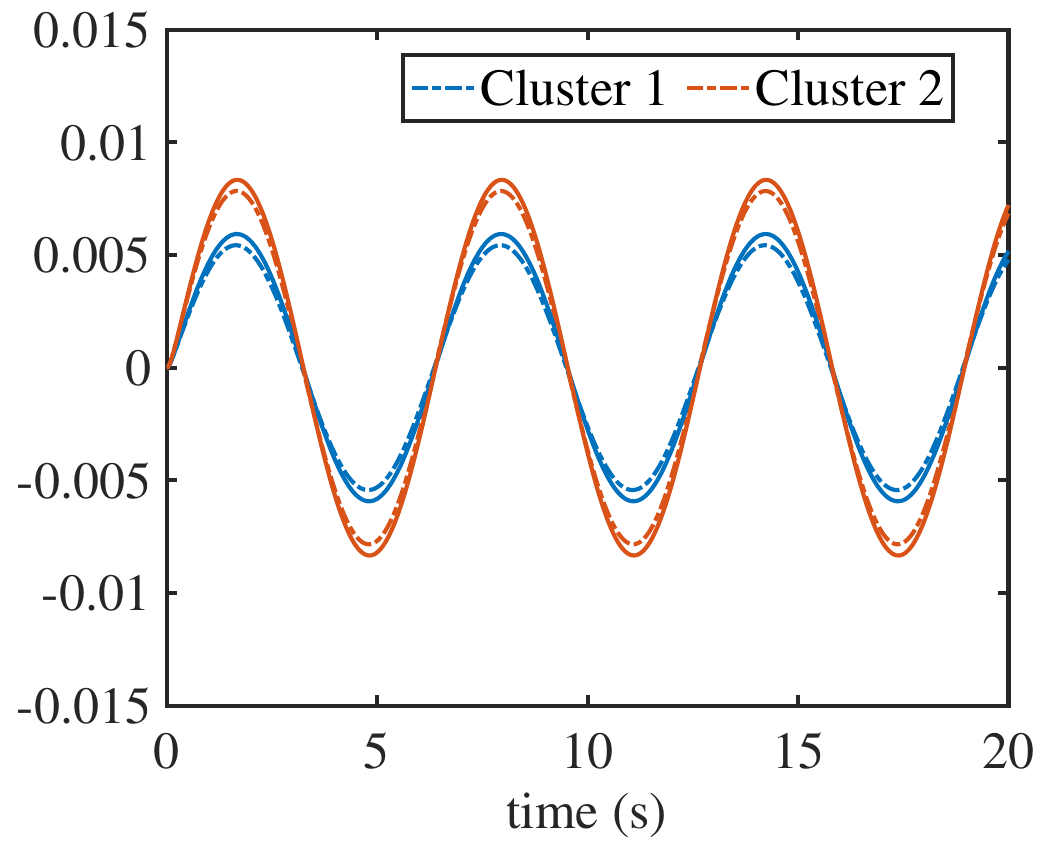}    
				\subcaption{}
				\label{fig:y}
			\end{minipage}%
			\begin{minipage}[t]{0.5\linewidth}
				\centering
				\includegraphics[width=\textwidth]{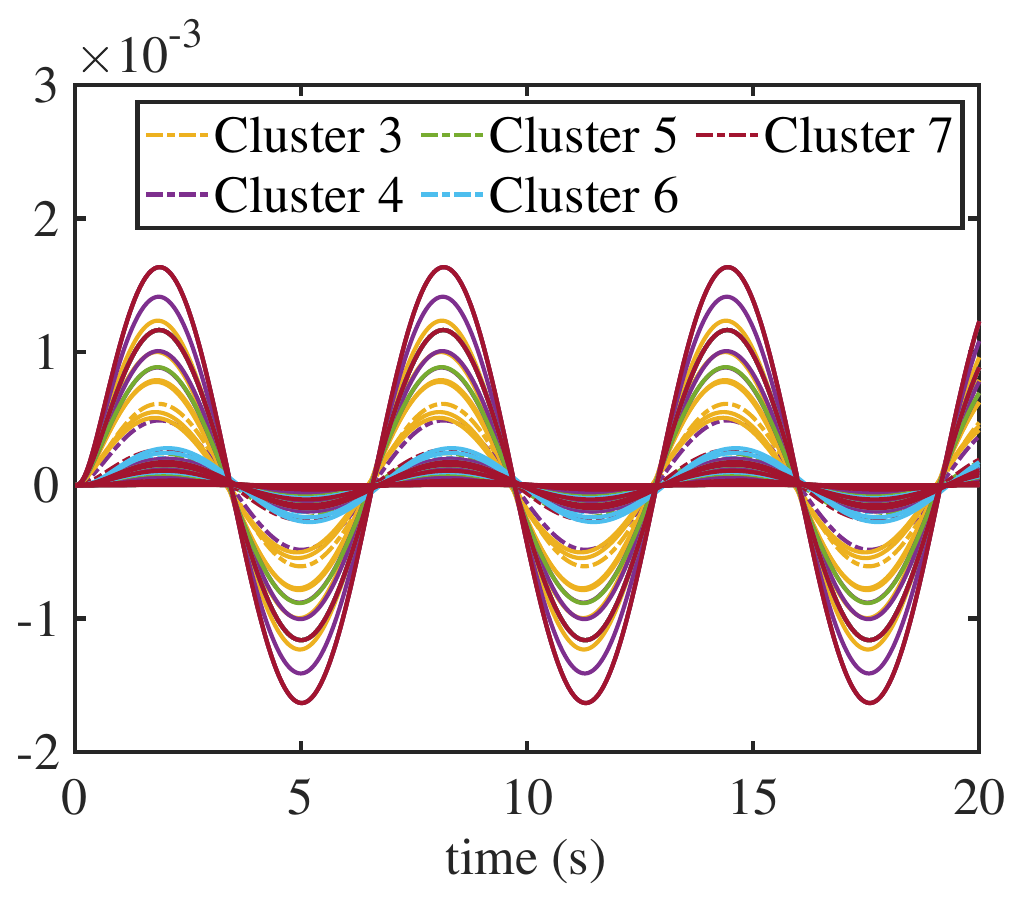}    
				\subcaption{}
				\label{fig:yr}
			\end{minipage}%
			\caption{(a) State trajectories of the original and reduced-order Lur'e network system for Clusters 1 and 2. (b) State trajectories of the original and reduced-order  Lur'e network system for Clusters 3 to 7. The solid curves represent the trajectories of the original network, and the dashed ones represent the trajectories of the reduced network.}
		\end{figure}

\section{Conclusions}
\label{sec:conclusion}


In this paper, we have introduced a clustering-based model reduction technique aimed at preserving the network structure of Lur'e network systems while ensuring the crucial property of absolute stability. We have provided the input-output error bound, which is determined by the $\mathcal{H}_\infty$-norm of a linear system parameterized in the characteristic matrix of a graph clustering.  
As for our future works, we will extend the proposed model reduction approach to  Lur'e network systems with non-scalar subsystems. Additionally, we will explore the preservation method of other important properties, such as synchronization and passivity.

\section*{Appendix}
\renewcommand{\thesubsection}{\Alph{subsection}}
\subsection{Proof of Lemma~\ref{Lemma3}}
\label{ap:lin_bound}
The transfer function of the error system for $u_e$ is
\begin{equation}
    g_{u_e}(s)-\hat{g}_{u_e}(s)=C_e\left(sI_{N+r}-A_e \right)^{-1}B_e,
\end{equation}
where $A_e=\begin{bmatrix}
    A_L&\\ & \hat{A}_L
\end{bmatrix}$, $B_e=\begin{bmatrix}
    F\\ \hat{F}
\end{bmatrix}$, $C_e=\begin{bmatrix}
    I &-\Pi
\end{bmatrix}$. 
Inspired by \cite{ishizaki2013model},
we introduce a pair of transfer matrices 
\begin{equation}\label{transfer matrices}
  T=\begin{bmatrix}
  -\Pi^{\dagger}&I_r\\I_N&0  
\end{bmatrix},~T^{-1}=\begin{bmatrix}
  0&I_N\\I_r&\Pi^{\dagger} 
\end{bmatrix}  
\end{equation}
such that
\begin{equation}\label{eq: err matrix}
\begin{aligned}
    \Tilde{A}_e&=TA_eT^{-1}=\begin{bmatrix}
        \Pi^{\dagger}A_L\Pi&\Pi^{\dagger}A_L(-I_N+\Pi \Pi^{\dagger})\\0&A_L
    \end{bmatrix},\\
    \Tilde{B}_e&=TB_e=\begin{bmatrix}
       0\\F
    \end{bmatrix},\\
    \Tilde{C}_e&=C_eT^{-1}=\begin{bmatrix}
        -\Pi&I_N-\Pi \Pi^{\dagger}
    \end{bmatrix}.
\end{aligned}
\end{equation}
Thus, from \eqref{eq: err matrix}, we have
\begin{equation}\label{gue error}
\begin{aligned}
   &g_{u_e}(s)-\hat{g}_{u_e}(s)=\Tilde{C}_e\left(sI_{N+r}-\Tilde{A}_e \right)^{-1}\Tilde{B}_e\\
   &=\left[\Pi(sI_r-\Pi^{\dagger}A_L\Pi)^{-1}\Pi^{\dagger}A_L+I_N\right](I_N-\Pi \Pi^{\dagger})\\
   &=H(s)g_{u_e}(s), 
\end{aligned}
\end{equation}
where $H(s)$ is shown as ($\ref{eq hs}$). Similarly, we obtain 
\begin{equation}\label{gv error}
\begin{aligned}
   g_v(s)-\hat{g}_v(s)
   =\Tilde{C}_{ve}\left(sI_{N+r}-\Tilde{A}_{ve} \right)^{-1}\Tilde{B}_{ve}=H(s)g_v(s), 
\end{aligned} 
\end{equation}
where $H(s)$ is shown as ($\ref{eq hs}$).

Then, we show that $\Pi^\dagger A_L\Pi$ is Hurwitz for any full rank characteristic matrix $\Pi$. Note that there exists a positive definite matrix $P_H: = \Pi^T K_\mu \Pi$ such that
\begin{align}
   &(\Pi^\dagger A_L\Pi)^T P_H + P_H \Pi^\dagger A_L\Pi 
   \nonumber\\
   =& \Pi^T A_L^T K_\mu  \Pi+ \Pi^T K_\mu A_L \Pi \nonumber\\
   =& \Pi^T (A_L^T K_\mu  + K_\mu A_L) \Pi  
   < 0.
\end{align}
Therefore, $H(s)$ is asymptotically stable, and its $\mathcal{H}_\infty$-norm is well defined.

Finally, according to (\ref{gue error}) and (\ref{gv error}), the inequality (\ref{e_gue}) and (\ref{e_gv}) hold.


\subsection{Proof of Theorem~\ref{thm:errbound}}
\label{ap:bound}
In the complex frequency domain, we have 
\begin{subequations}
\begin{align}
\label{eq:Y}
  X(s)&=g_{u_e}(s)U_e(s)+g_v(s)V(s), \\
  \hat{X}(s) & =\hat{g}_{u_e}(s)U_e(s)+\hat{g}_v(s)\hat{V}(s),
\end{align}
\end{subequations}
where $X(s)$ and $\hat{X}(s)$ are the Laplace transforms of the time domain signals $x(t)$ and $\hat{x}(t)$, respectively, assuming zero initial conditions. All the transfer matrices are defined in \eqref{eq:tfs}. Then, the approximation error in the complex frequency domain is given as 
 \begin{equation}
 \begin{aligned}
  &X(s)-\hat{X}(s)\\
  &=[g_{u_e}(s)-\hat{g}_{u_e}(s)]U_e(s)+[g_v(s)-\hat{g}_v(s)]V(s)\\
  &\quad +\hat{g}_v(s)[V(s)-\hat{V}(s)],
 \end{aligned}
 \end{equation}
 which leads to the following upper bound in the time domain:
\begin{equation}
\begin{aligned}
 \|x(t)-\hat{x}(t)\|_2  
   \leq  &\|g_{u_e}(s)-\hat{g}_{u_e}(s)\|_{\mathcal{H}_\infty}\|u_e(t)\|_2\\
  &+\|g_v(s)-\hat{g}_v(s)\|_{\mathcal{H}_\infty}\|v(t)\|_2\\  &+\|\hat{g}_v(s)\|_{\mathcal{H}_\infty}\|v(t)-\hat{v}(t)\|_2.  
\end{aligned}
\end{equation}
Then, we analyze each terms in the above error bound.

First, we discuss the bound of $\|v(t)-\hat{v}(t)\|_2$ and $\|v(t)\|_2$. Following a same procedure in \cite{ChengEJC201911}, we make use of the incremental sector bounded condition (\ref{incrementally sector bounded}) and obtain
\begin{equation*}
\begin{aligned}
& [v(t)-\hat{v}(t)]^T[v(t)-\hat{v}(t)]\\
&\leq [v(t)-\hat{v}(t)]^T[v(t)-\hat{v}(t)]\\
& \quad-[v(t)-\hat{v}(t)]^T[v(t)-\hat{v}(t)-K_\mu(\hat{x}(t)-x(t))]\\
&=[v(t)-\hat{v}(t)]^TK_\mu[\hat{x}(t)-x(t)]\\
       &\leq \frac{1}{2}[v(t)-\hat{v}(t)]^T[v(t)-\hat{v}(t)  \\
       &\quad +\frac{1}{2}[x(t)-\hat{x}(t)]^TK^2_\mu[x(t)-\hat{x}(t)].
\end{aligned}
\end{equation*}
Thus, $[v(t)-\hat{v}(t)]^T[v(t)-\hat{v}(t)] 
\leq \mu_{\text{max}}^2[x(t)-\hat{x}(t)]^T[x(t)-\hat{x}(t)], $
where $\mu_{\text{max}} > 0$ is the largest element of $K_\mu$. This also implies 
\begin{equation}\label{eq err_v}
  \|v(t)-\hat{v}(t)\|_2\leq \mu_{\text{max}}\|x(t)-\hat{x}(t)\|_2.
\end{equation}
Since $\Phi(0)=0$, which leads to
\begin{equation}
\begin{aligned}
 \|v(t)\|_2\leq& \mu_{\text{max}}\|x(t)\|_2 \leq  \mu_{\text{max}}\|g_{u_e}(s)\|_{\mathcal{H}_\infty}\|u_e(t)\|_2\\
 +&\mu_{\text{max}}\|g_v(s)\|_{\mathcal{H}_\infty}\|v(t)\|_2.   
\end{aligned}
\end{equation}
If $\|g_v(s)\|_{\mathcal{H}_\infty}<\frac{1}{\mu_{\text{max}}}$, then
\begin{equation}\label{eq v}
  \|v(t)\|_2\leq \frac{\mu_{\text{max}}\|g_{u_e}(s)\|_{\mathcal{H}_\infty}}{1-\mu_{\text{max}}\|g_v(s)\|_{\mathcal{H}_\infty}} \|u_e(t)\|_2.
\end{equation}
According to (\ref{e_gue}), it has
\begin{multline}
    \label{e_gue1}
    \|g_{u_e}(s)-\hat{g}_{u_e}(s)\|_{\mathcal{H}_\infty}\|u(t)\|_2\\ \leq \gamma_H\|g_{u_e}(s)\|_{\mathcal{H}_\infty}\|u(t)\|_2.
\end{multline}
By (\ref{eq v}) and (\ref{e_gv}), 
\begin{equation}\label{err v1}
\begin{aligned}
  &\|g_v(s)-\hat{g}_v(s)\|_{\mathcal{H}_\infty}\|v(t)\|_2 \\
  &\leq   \frac{\gamma_H \mu_{\text{max}}\|g_{u_e}(s)\|_{\mathcal{H}_\infty}\|g_v(s)\|_{\mathcal{H}_\infty}}{1-\mu_{\text{max}}\|g_v(s)\|_{\mathcal{H}_\infty}} \|u_e(t)\|_2.  
\end{aligned}
\end{equation}
According to (\ref{eq err_v}) and (\ref{e_gv}), it can be obtained that
\begin{equation}\label{hatgv}
\begin{aligned}
   &\|\hat{g}_v(s)\|_{\mathcal{H}_\infty}\|v(t)-\hat{v}(t)\|_2\\
   &\leq\|g_v(s)-\hat{g}_v(s)\|_{\mathcal{H}_\infty}\|v(t)-\hat{v}(t)\|_2\\
   &\quad +\|g_v(s)\|_{\mathcal{H}_\infty}\|v(t)-\hat{v}(t)\|_2 \\
   &\leq (\gamma_H+1)\mu_{\text{max}} \|g_v(s)\|_{\mathcal{H}_\infty}\|x(t)-\hat{x}(t)\|_2.
\end{aligned}
\end{equation}
 Thus, by (\ref{e_gue1}), (\ref{err v1}) and (\ref{hatgv}), we obtain
\begin{multline}
    \|x(t)-\hat{x}(t)\|_2 \leq (\gamma_H+1)\mu_{\text{max}} \|g_v(s)\|_{\mathcal{H}_\infty}\|x(t)-\hat{x}(t)\|_2\\
      +   \gamma_H\|g_{u_e}(s)\|_{\mathcal{H}_\infty} \|u_e(t)\|_2 \\ +
  \frac{\gamma_H \mu_{\text{max}}\|g_{u_e}(s)\|_{\mathcal{H}_\infty}\|g_v(s)\|_{\mathcal{H}_\infty}}{1-\mu_{\text{max}}\|g_v(s)\|_{\mathcal{H}_\infty}}   \|u_e(t)\|_2 
 .
\end{multline}
If $\|{g}_v(s)\|_{\mathcal{H}_\infty}<\frac{1}{\mu_{\text{max}}(\gamma_H+1)}$, then $\|{g}_v(s)\|_{\mathcal{H}_\infty}<\frac{1}{\mu_\text{max}}$ also holds. The error bound is obtained as
\begin{multline*} 
   \|x(t)-\hat{x}(t)\|_2 \\
  \leq \frac{\gamma_H\|g_{u_e}(s)\|_{\mathcal{H}_\infty}+
  \frac{\gamma_H \mu_{\text{max}}\|g_{u_e}(s)\|_{\mathcal{H}_\infty}\|g_v(s)\|_{\mathcal{H}_\infty}}{1-\mu_{\text{max}}\|g_v(s)\|_{\mathcal{H}_\infty}}}{1-(\gamma_H+1)\mu_{\text{max}} \|g_v(s)\|_{\mathcal{H}_\infty}} \|u_e(t)\|_2,
\end{multline*}
which can be simplified  to the inequality \eqref{errorbound} with $\Gamma(\gamma_H)$ defined in \eqref{Gamma}.
\bibliographystyle{IEEEtran}
\bibliography{ref}
\end{document}